\documentclass[smallextended,natbib,runningheads]{svjour3}

\smartqed  

\usepackage[utf8]{inputenc}
\usepackage{amsthm}
\usepackage{multirow}
\usepackage{amsfonts}
\usepackage{amsmath}
\usepackage{verbatim}
\usepackage{natbib}
\usepackage{xcolor}
\usepackage{mathtools}

\usepackage{setspace}

\title{Penalized Maximum Likelihood Estimator for Mixture of von Mises-Fisher Distributions}
\author{Tin Lok James Ng}

\institute{Tin Lok James Ng \at
              School of Computer Science and Statistics \\
              Trinity College Dublin, Ireland \\
              \email{ngja@tcd.ie}           
           }
\date{Received: date / Revised: date}

\begin{document}
\doublespacing
\maketitle

\begin{abstract}
The von Mises-Fisher distribution is one of the most widely used probability distributions to describe directional data. Finite mixtures of von Mises-Fisher distributions have found numerous applications. However, the likelihood function for the finite mixture of von Mises-Fisher distributions is unbounded and consequently the maximum likelihood estimation is not well defined. To address the problem of likelihood degeneracy, we consider a penalized maximum likelihood approach whereby a penalty function is incorporated. We prove strong consistency of the resulting estimator. An Expectation-Maximization algorithm for the penalized likelihood function is developed and experiments are performed to examine its performance.

\end{abstract}

\section{Introduction}
\label{sec:intro}
In many areas of statistical modelling, data are represented as directions or unit length vectors \citep{Mardia1972, Jupp1995, Mardia2000}. The analysis of directional data has attracted much research interests in various disciplines, from hydrology \citep{Chen2013} to biology \citep{Boomsma2006}, and from image analysis \citep{Zhe2019} to text mining \citep{Banerjee2005}. The von Mises-Fisher (vMF) distribution is one of the most commonly used distributions to model data distributed on the surface of the unit hypersphere \citep{Fisher1953, Mardia2000}. The vMF distribution has been applied successfully in many domains (e.g., \citep{Sinkkonen2002, Mcgraw2006, Bangert2010}).
\\\\
A mixture of vMF distributions \citep{Banerjee2005} assumes that each observation is drawn from one of the $p$ vMF distributions. Applications of the vMF mixture model are diverse, including image analysis \citep{Mcgraw2006} and text mining \citep{Banerjee2005}. More recently, it has been shown that the vMF mixture model can approximate any continuous density function on the unit hypersphere to arbitrary degrees of accuracy given sufficient numbers of mixture component \citep{Ng2020b}.
\\\\
Various estimation strategies have been developed to perform model estimation, including the maximum likelihood approach \citep{Banerjee2005} and Bayesian methods \citep{Bagchi1991, Taghia2014}. The maximum likelihood approach, which is typically performed using the Expectation-Maximization (EM) algorithm \citep{Dempster1977, Banerjee2005}, is among the most popular approach to parameter estimation. However, as we show in Section \ref{sec:likelihood} the likelihood function is unbounded from above, and consequently a global maximum likelihood estimate (MLE) fails to exist.
\\\\
The unboundedness of likelihood function occurs in various mixture modelling context, particularly for mixture models with location-scale family distributions including the mixture of normal distributions \citep{Ciuperca2003, Chen2008} and the mixture of Gamma distributions \citep{Chen2016}. Various approaches have been developed in order to tackle the likelihood degeneracy problem, including resorting to local estimates \citep{Peters1978}, compactification of the parameter space \citep{Redner1981}, and constrained maximization of the likelihood function \citep{Hathaway1985}. 
\\\\
An alternative solution to the likelihood degeneracy problem is to penalized the likelihood function such that the resulting penalized likelihood function is bounded and the existence of the penalized MLE is guaranteed. The approach of penalized maximum likelihood was applied to normal mixture models \citep{Ciuperca2003, Chen2008}, two-parameter Gamma mixture models \citep{Chen2016}. A penalized likelihood approach has also a Bayesian interpretation \citep{Goodd1971, Ciuperca2003}, whereby the penalized likelihood function corresponds to a posterior density and the penalized maximum likelihood solution to the maximum a posterior estimate.
\\\\
Previously, the penalized maximum likelihood approach is applied to study the mixture of von-Mises distributions \citep{Chen2007b} where consistency results were obtained. The von-Mises distribution is a special case of the von Mises-Fisher distribution defined on the circle. We generalize the results in \cite{Chen2007b} to the arbitrary dimensional sphere. The consistency proof in \cite{Chen2007b} relies heavily on the univariate properties of the von Mises distribution and generalization of the result to higher dimensions is not straightforward. In this paper we prove a few useful technical lemmas before proving the main results. To handle the non-identifiability of the mixture models, we use the framework of \cite{Redner1981} to obtain consistency in the quotient space. 
\\\\
In this paper, we also consider the penalized likelihood approach to tackle the problem of likelihood unboundedness for the mixture of vMF distributions. We incorporate a penalty term into the likelihood function and maximize the resulting penalized likelihood function. We study conditions on the penalty function to ensure consistency of the penalized maximum likelihood estimator (PMLE). We develop an Expectation-Maximization algorithm to perform model estimation based on the penalized likelihood function. The rest of the paper is structured as follows. Section \ref{sec:background} introduces the background on vMF mixtures and key notations used in the subsequent sections. The problem of likelihood degeneracy is formally presented in Section \ref{sec:likelihood}. Section \ref{sec:MLE} develops the penalized maximum likelihood approach and discusses conditions on the penalty function to ensure strong consistency of the resulting estimator. An Expectation-Maximization algorithm is also developed in Section \ref{sec:MLE}, and its performance is examined in Section \ref{sec:simulation}. Section \ref{sec:app} illustrate the proposed EM algorithm using a data application. We conclude the paper with a discussion section.

\section{Background}
\label{sec:background}
The probability density function of a $d$-dimensional vMF distribution is given by:
\begin{eqnarray}
\label{vmf_dens}
 f(\mathbf{x};\boldsymbol{\mu}, \kappa) = c_{d}(\kappa) e^{\kappa \boldsymbol{\mu}^{T} \mathbf{x}} ,
\end{eqnarray}
where $x \in \mathbb{S}^{d-1}$ is a $d$-dimensional unit vector (i.e.  $ \mathbb{S}^{d-1} := \{ \mathbf{x} \in \mathbb{R}^d: ||\mathbf{x}|| = 1 \}$ where $||\cdot||$ is the $L_2$ norm), $\boldsymbol{\mu} \in \mathbb{S}^{d-1}$ is the mean direction, and $\kappa \ge 0$ is the concentration parameter. The normalizing constant $c_d(\kappa)$ has the form
$$ c_d(\kappa) = \frac{\kappa^{d/2 - 1}}{(2 \pi)^{d/2} I_{d/2-1}(\kappa)},$$
where $I_r(\cdot)$ is the modified Bessel function of the first kind of order $r$. The vMF distribution becomes increasingly concentrated at the mean direction $\mu$ as the concentration parameter $\kappa$ increases. The case $\kappa=0$ corresponds to the uniform distribution on $\mathbb{S}^{d-1}$.  
\\
\\
The probability density function of a mixture of vMF distributions with $p$ mixture components can be expressed as
\begin{eqnarray}
\label{vmf_mix_dens}
  g(\mathbf{x}; \{\pi_k, \boldsymbol{\mu}_k, \kappa_k\}_{k=1}^{p}) = \sum_{k=1}^{p} \pi_k f(\mathbf{x}; \boldsymbol{\mu}_k, \kappa_k) ,
\end{eqnarray}
where $(\pi_1, \ldots, \pi_p)$ is the mixing proportions, $(\boldsymbol{\mu}_k, \kappa_k)$ are the parameters for the $k$th component of the mixture, and $f(\cdot; \boldsymbol{\mu}_k, \kappa_k)$ is the vMF density function defined in \eqref{vmf_dens}. 
\\\\
We let $\Theta := \{ \theta \equiv (\boldsymbol{\mu}, \kappa): \boldsymbol{\mu} \in \mathbb{S}^{d-1}, \kappa \ge 0 \}$ be the parameter space of the vMF distribution, with the metric $\rho( \cdot, \cdot)$ defined as
\begin{eqnarray}
\label{metric_theta}
 \rho(\theta_1, \theta_2) = \mbox{arccos}(\boldsymbol{\mu}_1^{T} \boldsymbol{ \mu}_2) + |\kappa_1 - \kappa_2| ,
\end{eqnarray}
for $\theta_1 = (\boldsymbol{\mu}_1, \kappa_1), \theta_2 = (\boldsymbol{\mu}_2, \kappa_2)$. For any $\theta = (\boldsymbol{\mu}, \kappa) \in \Theta$, we write $f_{\theta}(\cdot) := f(\cdot; \boldsymbol{\mu}, \kappa)$ for the density function and $\gamma_{\theta}$ for the corresponding measure. The space of mixing probabilities is denoted by $\Pi := \{ (\pi_1, \ldots, \pi_p): \sum_{i=1}^{p} \pi_p = 1, \pi_k \ge 0, k = 1, \ldots, p \}$.  A $p$-component mixture of vMF distributions can be expressed as $\gamma = \sum_{k=1}^{p} \pi_k \gamma_{\theta_k}$ where $(\pi_1, \ldots, \pi_p) \in \Pi$ and $(\theta_1, \ldots, \theta_p) \in \Theta^{p}$, and where $\Theta^{p} = \Theta \times \cdots \times \Theta$ is the product of the parameter spaces. We define the product space $\Gamma := \Pi \times \Theta^{p}$, and we slightly abuse notations to let $\gamma$ denote both the mixing measure and the parameters in $\Gamma$. While $\Gamma$ is a natural parameterization of the family of mixture of vMF distributions, elements of $\Gamma$ are not identifiable. Thus, we let $\tilde{\Gamma}$ be the quotient topological space obtained from $\Gamma$ by identifying all parameters $(\pi_1, \ldots, \pi_p, \theta_1, \ldots, \theta_p)$ such that their corresponding densities are equal (almost) everywhere. For the rest of the paper, we assume that the number of mixture components $p$ is known.

\section{Likelihood Degeneracy}
\label{sec:likelihood}
We investigate the likelihood degeneracy problem of the vMF mixture model in this section. For any observations generated from a vMF mixture model with two or more mixture components, we show that the resulting likelihood function on the parameter space $\Gamma$ is unbounded above. As discussed in Section \ref{sec:intro}, likelihood degeneracy is a common problem for mixture models with location-scale distributions, including the normal mixtures. In the case of normal mixture distributions, one can show that by letting the mean parameter of a mixture component equal to one of the observations and letting the variance of the same mixture component converge to zero while holding other parameters fixed, the likelihood function diverges to positive infinity \citep{Chen2008}.
\\\\
For the vMF mixture distributions, the likelihood unboundedness can be best understood in the special case of $x \in \mathbb{S}^1$, or the mixture of von Mises distributions. The von Mises distribution, also known as the circular normal distribution, approaches a normal distribution with large concentration parameter $\kappa$:
$$ f(x|\mu, \kappa) \approx \frac{1}{\sigma \sqrt{2 \pi}} \exp \bigg[ \frac{-(x - \mu)^2}{2 \sigma^2} \bigg] ,$$
with $\sigma^2 = 1 / \kappa$, and the approximation converges uniformly as $\kappa$ goes to infinity. Therefore, the likelihood function of a mixture of von Mises distributions diverges to infinity by letting the mean parameter of a mixture component equal to one of the observations and letting the concentration parameter diverges to infinity. 
\\\\
We now consider the general case of the vMF mixture models. Let ${\cal X} = \{\mathbf{x}_1, \ldots, \mathbf{x}_n\}$ be the observations generated from a mixture of vMF distributions with density function $\sum_{k=1}^{p} \pi_k f_{\theta_k}(\cdot)$ where $\theta_k = (\boldsymbol{\mu}_k, \kappa_k)$. The likelihood function can be expressed as:  
$$ L({\cal X}; \boldsymbol{\theta}, \boldsymbol{\pi}) = \prod_{i=1}^{n} \sum_{k=1}^{p} \pi_k f(\mathbf{x}_i; \boldsymbol{\mu}_k, \kappa_k) ,$$ 
where $\boldsymbol{\theta} = (\theta_1, \ldots, \theta_p) = ((\boldsymbol{\mu}_1, \kappa_1), \ldots, (\boldsymbol{\mu}_p, \kappa_p))$ and $\boldsymbol{\pi} = (\pi_1, \ldots, \pi_p)$. We can show that by letting the mean direction $\boldsymbol{\mu}_k$ of one of the mixture components equals to an arbitrary observation and letting the corresponding concentration parameter $\kappa_k$ goes to infinity, the resulting likelihood function diverges.

\begin{theorem}
\label{thm:divg}
For any observations ${\cal X} = (\mathbf{x}_1, \ldots, \mathbf{x}_n)$, there exists a sequence $(\boldsymbol{\theta}^{(q)}, \boldsymbol{\pi}^{(q)})$, $q=1,2,\ldots$ such that $L({\cal X}; \boldsymbol{\theta}^{(q)}, \boldsymbol{\pi}^{(q)}) \uparrow \infty$ as $q \rightarrow \infty$.
\end{theorem}
The proof of Theorem~\ref{thm:divg} is provided in the Appendix. The unboundedness of the likelihood function on the parameter space implies that the maximum likelihood estimator is not well defined. 

\section{Penalized Maximum Likelihood Estimation}
\label{sec:MLE}
\subsection{Preliminary}
Let $\gamma_0 \in \Gamma$ be the true mixing measure for the mixture of vMF distributions with corresponding density function $f_0$ on $\mathbb{S}^{d-1}$. We let $M$ be the maximum of the true density $f_0$:
\begin{eqnarray}
\label{M_def}
M := \max_{\mathbf{x} \in \mathbb{S}^{d-1}} f_0(\mathbf{x}) ,
\end{eqnarray}
and define the metric $d(\mathbf{x}, \mathbf{y}) = \arccos(\mathbf{x}^{T} \mathbf{y})$ on $\mathbb{S}^{d-1}$ as the angle between two unit vectors $\mathbf{x}, \mathbf{y} \in \mathbb{S}^{d-1}$.
For any fixed $\mathbf{x} \in \mathbb{S}^{d-1}$ and positive number $\epsilon$, the $\epsilon$-ball in $\mathbb{S}^{d-1}$ centered at $\mathbf{x}$ is defined as $B_{\epsilon}(\mathbf{x}) = \{\mathbf{y} \in \mathbb{S}^{d-1}: d(\mathbf{x}, \mathbf{y}) < \epsilon\}$. For any measurable set $B \subset \mathbb{S}^{d-1}$, the spherical measure of $B$ is given by $ \omega(B) := \int_{B} d \omega $, where $d \omega$ is the standard surface measure on $\mathbb{S}^{d-1}$. 
\\\\
For any $\mathbf{x} \in \mathbb{S}^{d-1}$ and small positive number $\epsilon$, the measure of the ball $B_{2\epsilon}(\mathbf{x})$ in $\mathbb{S}^{d-1}$ is given by \citep{Li2011}
\begin{eqnarray}
\omega(B_{2 \epsilon}(\mathbf{x})) &=& \frac{2 \pi^{(d-1)/2}}{\Gamma(\frac{d-1}{2})} \int_{0}^{2\epsilon} \sin^{d-2}(\theta) d\theta \nonumber \\
&\le& 2^{d-1} \frac{2 \pi^{(d-1)/2}}{\Gamma(\frac{d-1}{2})} \epsilon^{d-1}  \\
&=& A_2 \epsilon^{d-1},
\end{eqnarray}
where 
\begin{eqnarray}
\label{A2_def}
 A_2 = 2^{d-1} \frac{2 \pi^{(d-1)/2}}{\Gamma(\frac{d-1}{2})} .
\end{eqnarray}
We define the function $\delta(\cdot)$ by 
\begin{eqnarray}
\label{delta_def}
\delta(\epsilon) := M A_2 \epsilon^{d-1} ,
\end{eqnarray} 
where the constants $M$ and $A_2$ are defined in Equation \eqref{M_def} and \eqref{A2_def}, respectively. The function $\delta(\cdot)$ plays a crucial role in Lemmas \ref{key_lemma_1} and \ref{key_lemma_2}. Lemmas \ref{key_lemma_1} and \ref{key_lemma_2} are analogous to Lemmas 1 and 2 in \cite{Chen2008}. They provide (almost sure) upper bounds on the number of observations in a small $\epsilon$-ball in $\mathbb{S}^{d-1}$. The upper bound in Lemma \ref{key_lemma_1} is for each fixed $\epsilon$ in an interval whereas the upper bound in Lemma \ref{key_lemma_2} holds uniformly for all $\epsilon$ in the same interval. The proof of Lemma \ref{key_lemma_1} is given in the Appendix. The proof of Lemma \ref{key_lemma_2} is similar to the proof of Lemma 2 in \cite{Chen2008} and is omitted. Lemmas \ref{key_lemma_1} and \ref{key_lemma_2} are crucial to ensure consistency of the penalized maximum likelihood estimator. 
\\\\
We note that Lemmas \ref{key_lemma_1} and \ref{key_lemma_2} may be generalized by relaxing the assumption that the true density is a mixture of vMF densities. This is possible because the vMF assumption does not play a crucial role. Such a generalization has been obtained for normal mixtures \citep[Lemma 3.2]{Chen2017}. However, this is not required for the proof of our main result.

\begin{lemma}
\label{key_lemma_1}
For any sufficiently small positive number $\xi_0$, as $n \rightarrow \infty$, and for each fixed $\epsilon$ such that $$\frac{\log n}{M n A_2} \le \epsilon^{d-1} < \xi_0 ,$$ the following inequalities hold except for a zero probability event:
\begin{eqnarray}
\label{ineq1}
\sup_{\boldsymbol{\mu} \in S^{d-1}} \bigg\{ \frac{1}{n} \sum_{i=1}^{n} I\big(X_i \in B_{\epsilon}(\boldsymbol{\mu})\big)  \bigg\} \le 2 \delta(\epsilon).
\end{eqnarray}  
Uniformly for all $\epsilon$ such that $$ 0 < \epsilon^{d-1} < \frac{\log n}{M n A_2}, $$ the following inequalities hold except for a zero probability event:
\begin{eqnarray}
\label{ineq2}
 \sup_{\boldsymbol{\mu} \in S^{d-1}} \bigg\{ \frac{1}{n} \sum_{i=1}^{n} I\big(X_i \in B_{\epsilon}(\boldsymbol{\mu})\big)  \bigg\} \le 2 \frac{(\log n)^2}{n} .
\end{eqnarray}
\end{lemma}

\begin{lemma}
\label{key_lemma_2}
For any sufficiently small positive number $\xi_0$, as $n \rightarrow \infty$, uniformly for all $\epsilon$ such that
$$\frac{\log n}{M n A_2} \le \epsilon^{d-1} < \xi_0 ,$$
the following inequality holds except for a zero probability event:
\begin{eqnarray}
\label{ineq1}
\sup_{\boldsymbol{\mu} \in S^{d-1}} \bigg\{ \frac{1}{n} \sum_{i=1}^{n} I\big(X_i \in B_{\epsilon}(\boldsymbol{\mu})\big)  \bigg\} \le 4 \delta(\epsilon).
\end{eqnarray}  
Uniformly for all $\epsilon$ such that $$ 0 < \epsilon^{d-1} < \frac{\log n}{M n A_2}, $$ the following inequalities hold except for a zero probability event:
\begin{eqnarray}
\label{ineq2}
 \sup_{\boldsymbol{\mu} \in S^{d-1}} \bigg\{ \frac{1}{n} \sum_{i=1}^{n} I\big(X_i \in B_{\epsilon}(\boldsymbol{\mu})\big)  \bigg\} \le 2 \frac{(\log n)^2}{n} .
\end{eqnarray}
\end{lemma}

\subsection{Penalized Maximum Likelihood Estimator}
For any mixing measure of a $p$-component mixture $\gamma = \sum_{l=1}^{p} \pi_l \gamma_{\theta_l} $ in $\Gamma$, and $n$ i.i.d. observations ${\cal X}$, the penalized log-likelihood function is defined as
\begin{eqnarray}
\label{pen_lik}
 pl_n(\gamma) = l_n(\gamma) + p_n(\boldsymbol{\kappa}) 
\end{eqnarray}
where $l_n(\gamma)$ is the log-likelihood function:
$$ l_n(\gamma) = \sum_{i=1}^{n} \log \bigg\{ \sum_{k=1}^{p} \pi_k f(\mathbf{x}_i; \boldsymbol{\mu}_k, \kappa_k) \bigg\}, $$
and $p_n(\cdot)$ is a penalty function that depends on $\boldsymbol{\kappa} = (\kappa_1, \ldots, \kappa_p)$. Note that we slightly abuse notations and let $p_n(\cdot)$ denotes the penalty function and $p$ denotes the number of mixture components. We impose the following conditions on the penalty function $p_n(\cdot)$.
\begin{enumerate}
    \item[C1] $p_n(\boldsymbol{\kappa}) = \sum_{l=1}^{p} \tilde{p}_n(\kappa_l)$,
    \item[C2] For $l=1,\ldots, p$, $\sup_{\kappa_l > 0} \max \{0, \tilde{p}_n(\kappa_l) \} = o(n)$ and $\tilde{p}_n(\kappa_l) = o(n)$ for each fixed $\kappa_l \ge 0$,
    \item[C3] For $l=1, \ldots, p$, and for $$ 0 < \frac{1}{\log (\kappa_l)^{2d - 2}} \le \frac{\log n}{M n A_2} ,$$ $\tilde{p}_n(\kappa_l) \le -3 (\log n)^{2} \log \kappa_l$ for large enough $n$.
\end{enumerate}
Conditions C1 - C3 on the penalty function are analogous to the three conditions proposed in \citep{Chen2008}. Condition C1 assumes that the penalty function is of additive form. Condition C2 ensures that the penalty is not overly strong while condition C3 allows the penalty to be severe when the concentration parameter is very large. Recall the true mixing measure $\gamma_0 \in \Gamma$, and let $\hat{\gamma}$ denote the maximizer of the penalized log-likelihood function defined in Equation \eqref{pen_lik}. We have the following main result of this paper demonstrating that the maximizer of the penalized log-likelihood function is strongly consistency.

\begin{theorem}
Let $\hat{\gamma}_n$ be the maximizer of the penalized log-likelihood $pl_n(\gamma)$, then $\hat{\gamma}_n \rightarrow \gamma_0$ almost surely in the quotient topological space $\tilde{\Gamma}$.
\end{theorem}

\subsection{EM Algorithm}
\label{subsec:EM}
We develop an Expectation-Maximization algorithm to maximize the penalized log-likelihood function defined in Equation \eqref{pen_lik}. By condition C1, the penalty function is assumed to have the form $p_n(\boldsymbol{\kappa}) = \sum_{l=1}^{p} \tilde{p}_n(\kappa_l)$. We consider $\tilde{p}_n(\kappa_l)$ to have the form $\tilde{p}_n(\kappa_l) = - \psi_n \kappa_l$ for all $l$ where the constant $\psi_n \propto n^{-1}$ that depends on the sample size $n$. In particular, we may set $\psi_n = \zeta / n$ for some constant $\zeta > 0$ or $\psi_n = S_x / n$ where $S_x$ is the sample circular variance.
\\\\
The resulting penalty function clearly satisfies condition C2. We note that condition C3 is also satisfied since for  $$ 0 < \frac{1}{\log (\kappa_l)^{2d - 2}} \le \frac{\log n}{M n A_2} ,$$ we have $$ \kappa_l \approx \exp \big( (n / \log n)^{1/(2d-2)} \big) .$$
The EM algorithm developed in \citep{Banerjee2005} can be easily modified to incorporate an additional penalty function. The E-Step of the penalized EM involves computing the conditional probabilities:
\begin{eqnarray}
\label{estep}
 p(Z_i = h|\mathbf{x}_i, \boldsymbol{\theta}) = \frac{\pi_h f(\mathbf{x}_i; \theta_h)}{ \sum_{l=1}^{p} \pi_l f(\mathbf{x}_i;\theta_l) }, \quad h=1,\ldots, p,
\end{eqnarray}
where $Z_i$ is the latent variable denoting the cluster membership of the $i$th observation. For the M-step, using the method of Lagrange multipliers, we optimize the full conditional penalized log-likelihood function below  
$$ \sum_{l=1}^{p} \bigg[ \sum_{i=1}^{n} (\log( \pi_l ) + \log(c_d(\kappa_l)) ) p(Z_i=l| \mathbf{x}_i, \boldsymbol{\theta}) + \sum_{i=1}^{n} \kappa_l \mu_l^{T} x_i p(Z_i=l| \mathbf{x}_i, \boldsymbol{\theta}) - \psi_n \kappa_l  + \lambda_l (1 - \boldsymbol{\mu}_l^{T} \boldsymbol{\mu}_l ) \bigg] $$ 
with respect to $\boldsymbol{\mu}_h, \kappa_h, \pi_h$ for $h=1,\ldots,p$, which gives:
\begin{eqnarray}
 \hat{\pi}_h &=& \frac{1}{n} \sum_{i=1}^{N} p(Z_i = h|\mathbf{x}_i, \boldsymbol{\theta}) \\
  \hat{\boldsymbol{\mu}}_h &=& \frac{ r_h }{ ||r_h||} \\
  \label{mstep_kappa}
 \frac{I_{d/2}(\hat{\kappa}_h)}{I_{d/2-1}(\hat{\kappa}_h)} &=& \frac{- \psi_n + ||r_h||}{\sum_{i=1}^{N} p(Z_i=h|\mathbf{x}_i, \boldsymbol{\theta})} 
\end{eqnarray}
where $ r_h =  \sum_{i=1}^{n} \mathbf{x}_i p(Z_i = h|\mathbf{x}_i, \boldsymbol{\theta})$. We note that the assumption on $\psi_n$ implies that $-\psi_n + ||r_h|| \ge 0$ almost surely as $n \rightarrow \infty$. However, for a finite sample size, there is a non-zero possibility that
$ -\psi_n + ||r_h|| < 0$, and the updating equation for $\kappa_h$ is not well defined since the left hand side of Equation \eqref{mstep_kappa} is non-negative. However, the left hand side of Equation \eqref{mstep_kappa} is a strictly monotonically increasing function from $[0, \infty)$ to $[0, 1)$ \citep{Schou1978, Hornik2014}, and in particular $\hat{\kappa}_h = 0$ whenever  
$$\frac{I_{d/2}(\hat{\kappa}_h)}{I_{d/2-1}(\hat{\kappa}_h)} = 0. $$
Thus, we can simply set $\kappa_h = 0$ whenever $ -\psi_n + ||r_h|| < 0$. To solve Equation \eqref{mstep_kappa} for $\hat{\kappa}_h$, various approximations have been proposed \citep{Banerjee2005, Tanabe2007, Song2012}. Section 2.2 of \cite{Hornik2014} contains a detailed review of available approximations. We consider the approximation used in \cite{Banerjee2005}:
$$ \hat{\kappa}_h \approx \frac{\rho_h (d - \rho_h^2)}{ 1 - \rho_h^{2}} ,$$
with $$ \rho_h = \frac{- \psi_n + ||r_h||}{\sum_{i=1}^{N} p(Z_i=h|\mathbf{x}_i, \boldsymbol{\theta})} .$$
We initialize the EM algorithm by randomly assigning the observations into mixture components, and the algorithm is terminated if the change in the penalized log-likelihood falls below a small threshold which is set at $10^{-5}$ in the experiements.

\section{Simulation Studies}
\label{sec:simulation}
We perform simulation studies to investigate the performance of the proposed EM algorithm for maximizing the penalized likelihood function. We generate data from the mixture of vMF distributions with two and three mixture components and with dimensions $d=2,3,4$. For each model, data are generated with increasing samples sizes to assess the convergence of the estimated parameters toward the true parameters. The concentration parameters $\boldsymbol{\kappa}$ and the mixing proportions $\boldsymbol{\pi}$ are pre-specficied whereas the mean directions $\boldsymbol{\mu}$ are drawn from the uniform distribution on the surface of the unit hypersphere. 
\\\\
For the two mixture components model, we specify the mixing proportions as $\boldsymbol{\pi} = (0.5, 0.5)$ and the concentration parameters $\boldsymbol{\kappa} = (10, 1)$. For the model with three mixture components, we set $\boldsymbol{\pi} = (0.4, 0.3, 0.3)$ and $\boldsymbol{\kappa} = (10, 5, 1)$. For illustrative purpose, we consider the penalty function $\tilde{p}_n(\kappa_l) = - (1/n) \kappa_l$. For each combination of dimension $d$ and sample size $n$, we simulate 500 random samples from the model and the EM algorithm developed in Section \ref{subsec:EM} is used to obtain the parameter estimates. We measure the distance between the estimated parameters and the true parameters for each random sample. For the mean direction parameters $\boldsymbol{\mu}$, the distance is measured using the metric $d(\mathbf{x}, \mathbf{y}) = \mbox{arccos}(\mathbf{x}^T \mathbf{y})$.
\\\\
Simulation results for the two and three mixture cases are presented in Table \ref{conv_ana_1} and \ref{conv_ana_2}, respectively. The average distance and the standard deviation between the true and the estimated parameters from 500 replications are reported. We observe that the estimated parameters converge to the true parameter as $n$ increases. We notice that the mean direction parameter can be estimated with higher precision when the corresponding concentration parameter is large. This is expected since observations are more closely clustered with a large concentration parameter. 
\\\\
Table \ref{degeneracy1} and \ref{degeneracy2} show the number of degeneracies when running the EM algorithm for computing the ordinary MLE for mixture of vMF distributions. Observations are generated from mixture of vMF distributions with one mixture component for Table \ref{degeneracy1} and with two mixture components for Table \ref{degeneracy2}. We vary the dimension of the data from $d=3$ to $d=4$ and the sample size from $n=100$ to $n=500$. Mixtures of vMF distributions with $p=2,3,4,5$ components are fitted to the generated data. We compute the ordinary MLE using the EM algorithm and record the number of times that the EM fails to converge from 1000 simulation runs. The EM algorithm is considered fail to converge if one of the concentration parameters becomes exceedingly large (greater than $10^{10}$). From Table \ref{degeneracy1} and \ref{degeneracy2}, the EM algorithm tends to fail to converge with smaller sample sizes. We also note that when the fitted model has a larger number of mixture components $p$, the EM algorithm is more likely to fail to converge.

\begin{table}[hbtp]
\caption{Simulation results for the vMF mixtures with two mixture components. }
\begin{center}
\begin{tabular}{ccccccc}
   \hline
$d$ & $n$ & $\pi_1 (=0.5)$ & $\mu_1$ & $\mu_2$ & $\kappa_1(=10)$ & $\kappa_2(=1)$   \\ \hline\hline
\multirow{2}{*}{2} & \multirow{2}{*}{100} & 0.047 & 0.035 & 0.152 & 2.488 & 0.207 \\
 & & (0.050) & (0.023) & (0.124) & (2.339) & (0.159) \\ 
\\
\multirow{2}{*}{2} & \multirow{2}{*}{500} & 0.026 & 0.016 & 0.071 &1.594 &0.081 \\
& & (0.024) & (0.010) & (0.062) & (1.181) & (0.064) \\
\\
\multirow{2}{*}{2} & \multirow{2}{*}{1000} & 0.022 & 0.010  & 0.046 & 1.410 & 0.078 \\
& & (0.019) & (0.007) & (0.034) &  (1.037) & (0.075) \\ 
\\
\multirow{2}{*}{3} & \multirow{2}{*}{100} & 0.037 &0.048 & 0.275 &2.175  &0.171 \\
&  & (0.034) & (0.026) & (0.154) & (1.712) & (0.143) \\
\\
\multirow{2}{*}{3} & \multirow{2}{*}{500} & 0.025&0.023  &0.126 &1.345  &0.098  \\
& &  (0.024) & (0.013) & (0.067) & (0.894) & (0.087) \\
\\
\multirow{2}{*}{3} & \multirow{2}{*}{1000} & 0.022 &0.018  &0.085 &1.299  & 0.068 \\
& & (0.017) & (0.009) & (0.047) & (0.680) & (0.058)\\
\\
\multirow{2}{*}{4} & \multirow{2}{*}{100} & 0.039  &0.075  &0.324  &1.623  & 0.194  \\
& & (0.033) & (0.032) & (0.229) & (1.406) & (0.122) \\
\\
\multirow{2}{*}{4} & \multirow{2}{*}{500} &0.019  &0.024 &0.161 &0.868 & 0.103 \\
&& (0.017) & (0.013) & (0.065) & (0.518) & (0.083) \\
\\
\multirow{2}{*}{4} & \multirow{2}{*}{1000} & 0.018  & 0.020  &0.142  &0.842  &0.060  \\
& & (0.011) &  (0.011) & (0.052) & (0.431) & (0.051) \\
 \hline\hline

\end{tabular}
\end{center}
\label{conv_ana_1}
\end{table}

\begin{table}[hbtp]
\caption{Simulation results for the vMF mixtures with three mixture components}
\begin{center}
\begin{tabular}{cccccccccc}
   \hline
$d$ & $n$ & $\pi_1 (=0.4))$ & $\pi_2 (=0.3)$ & $\mu_1$ & $\mu_2$ & $\mu_3$ & $\kappa_1(=10)$ & $\kappa_2(=5)$ & $\kappa_3 (=1)$   \\ \hline\hline
\multirow{2}{*}{2} & \multirow{2}{*}{100} & 0.071 & 0.039 & 0.046 & 0.085 & 0.327 & 2.828 & 2.016 & 0.293 \\
 & & (0.042) & (0.026) & (0.050) & (0.091) & (0.279) & (2.571) & (1.289) & (0.302) \\
\\
\multirow{2}{*}{2} & \multirow{2}{*}{500} &0.058 & 0.028 & 0.039 & 0.062 & 0.209 & 1.703 & 1.514 & 0.255 \\
 & & (0.501) & (0.023) & (0.044) & (0.061) & (0.154) & (1.616) & (1.176) & (0.202) \\
\\
\multirow{2}{*}{2} & \multirow{2}{*}{1000} & 0.046 & 0.025 & 0.022 & 0.040 & 0.167 & 1.431 & 1.318 & 0.209 \\
 & & (0.032) & (0.185) & (0.036) & (0.047) & (0.125) & (1.307) & (0.892) & (0.185) \\
\\
\multirow{2}{*}{3} & \multirow{2}{*}{100} & 0.037 & 0.041 & 0.053 & 0.113 & 0.452 & 1.717 & 1.720 & 0.249 \\
&& (0.034) & (0.044) & (0.028) & (0.096) & (0.258) & (1.224) & (1.050) & (0.274) \\
\\
\multirow{2}{*}{3} & \multirow{2}{*}{500} & 0.033 & 0.031 & 0.043 & 0.067 & 0.285 & 1.120 & 1.018 & 0.206 \\
&& (0.024) & (0.023) & (0.037) & (0.040) & (0.138) & (1.010) & (0.914) & (0.246) \\
\\
\multirow{2}{*}{3} & \multirow{2}{*}{1000} & 0.026 & 0.022 & 0.024 & 0.052 & 0.255 & 1.051 & 1.039 & 0.183 \\
&& (0.026) & (0.021) & (0.018) & (0.029) & (0.126) & (0.806) & (0.747) & (0.138) \\
\\
\multirow{2}{*}{4} & \multirow{2}{*}{100} & 0.051 & 0.021 & 0.073 & 0.121 & 0.417 & 1.432 & 1.356 & 0.334 \\
&& (0.045) & (0.017) & (0.026) & (0.058) & (0.267) & (1.207) & (1.110) & (0.260) \\
\\
\multirow{2}{*}{4} & \multirow{2}{*}{500} & 0.030 & 0.022 & 0.031 & 0.068 & 0.313 & 1.154 & 1.088 & 0.246 \\
&& (0.026) & (0.016) & (0.016) & (0.028) & (0.018) & (0.873) & (0.760) & (0.209) \\  
\\
\multirow{2}{*}{4} & \multirow{2}{*}{1000} & 0.033 & 0.021 & 0.028 & 0.059 & 0.277 & 1.100 & 1.072 & 0.227 \\
&& (0.027) & (0.017) & (0.015) & (0.029) & (0.163) & (0.675) & (0.736) & (0.180) \\
\\
 \hline\hline

\end{tabular}
\end{center}
\label{conv_ana_2}
\end{table}

\begin{table}[hbtp]
\caption{Number of degeneracies of the EM algorithm when computing the ordinary MLE for mixture of vMF distributions with one mixture component}
\begin{center}
\begin{tabular}{cccccc}
   \hline
$d$ & $n$ & $p=2$ & $p=3$ & $p=4$ & $p=5$   \\ \hline\hline
3 & 100  & 13 & 47 &86 & 218 \\
3 & 200 & 2 &8 &34 & 53 \\
3 & 500 &0 &2 &2 & 5 \\
4 & 100  & 6 & 33 & 64 & 169 \\
4 & 200 & 1& 3& 12 & 39 \\
4 & 500 & 0& 0& 1&4 \\
 \hline\hline
\end{tabular}
\end{center}
\label{degeneracy1}
\end{table}

\begin{table}[hbtp]
\caption{Number of degeneracies of the EM algorithm when computing the ordinary MLE for mixture of vMF distributions with two mixture components}
\begin{center}
\begin{tabular}{cccccc}
   \hline
$d$ & $n$ & $p=2$ & $p=3$ & $p=4$ & $p=5$   \\ \hline\hline
3 & 100  & 0 & 51 &147 & 233  \\
3 & 200 &  0& 27 &45 & 87  \\
3 & 500 & 0& 4 & 11& 23 \\
4 & 100  &0 & 49 & 113 & 193 \\
4 & 200 & 0& 14 & 47 & 81  \\
4 & 500 & 0& 2& 8& 11 \\
 \hline\hline
\end{tabular}
\end{center}
\label{degeneracy2}
\end{table}

\section{Data Application}
\label{sec:app}
We illustrate the EM algorithm for maximum penalized log-likelihood using the \emph{household} data set from \textbf{R} package \textbf{HSAUR3}. The data set contains the household expenditures of 20 single men and 20 single women on four commodity group. As in  \cite{Hornik2014}, we will also focus on the three commodity groups (housing, food and service). The EM algorithms for ordinary MLE and for the penalized MLE with 2 and 3 mixture components are fitted to the data. The results are shown in Table \ref{housing_mle} and \ref{housing_pmle}, respectively, where the estimated parameters $\hat{\boldsymbol{\pi}}, \hat{\boldsymbol{\mu}}, \hat{\boldsymbol{\kappa}}$ are shown for all cases. The estimated paramters for the MLE and for the penalized MLE are very similar for both $p=2$ and $p=3$. The log-likelihood evaluated at the MLE is slightly larger than the penalized log-likelihood evaluated at the penalized MLE. More interestingly, we observe that for each case the largest concentration parameter obtained under the penalized MLE is smaller than that obtained under the MLE. This behavior suggests that the incorporation of a penalty function pulls back the estimate of largest concentration parameter towards 0 and prevents the divergence of the likelihood function. 

\begin{table}[hbtp]
\caption{Maximum likelihood estimates obtained from fitting mixtures of vMF distributions to the household expenses example.}
\begin{center}
\begin{tabular}{ccccccc}
   \hline
 & $\pi$ & housing & food & service & $\kappa$  & log likelihood   \\ \hline\hline
\multirow[t]{2}{*}{$p = 2$} & 0.47  & 0.95 & 0.13 & 0.27 & 114.70 & 113.08  \\
 & 0.53 &  0.67 & 0.63 & 0.40 & 17.96  & \\
\multirow[t]{3}{*}{$p = 3$} & 0.52 & 0.95 & 0.15 & 0.27 & 83.26 & 126.06 \\
 & 0.13 &0.67 & 0.31 & 0.68 & 181.21 & \\
 & 0.35 & 0.59& 0.76 & 0.28 & 62.91  & \\
 \hline\hline
\end{tabular}
\end{center}
\label{housing_mle}
\end{table}

\begin{table}[hbtp]
\caption{Penalized maximum likelihood estimates obtained from fitting mixtures of vMF distributions to the household expenses example.}
\begin{center}
\begin{tabular}{ccccccc}
   \hline
 & $\pi$ & housing & food & service & $\kappa$  & penalized log likelihood  \\ \hline\hline
\multirow[t]{2}{*}{$p = 2$} & 0.47  & 0.95 & 0.13 & 0.27 & 112.20 & 112.94 \\
 & 0.53 &  0.67 & 0.63 & 0.40 & 18.48  & \\
\multirow[t]{3}{*}{$p = 3$} & 0.52 & 0.95 & 0.15 & 0.27 & 82.97 & 125.26\\
 & 0.13 &0.67 & 0.31 & 0.68 & 165.71 & \\
 & 0.35 & 0.59& 0.76 & 0.28 & 62.69  & \\
 \hline\hline
\end{tabular}
\end{center}
\label{housing_pmle}
\end{table}

\section{Discussion}

In this paper we considered a penalized maximum likelihood approach to the estimation of the mixture of vMF distributions. By incorporating a suitable penalty function, we showed that the resulting penalized MLE is strongly consistent. An EM algorithm was derived to maximize the penalized likelihood function, and its performance and behavior were examined using simulation studies and a data application. The techniques used in this work to prove consistency could be applicable to study other mixture models for spherical observations.

\section{Conflict of Interests}
On behalf of all authors, the corresponding author states that there is no conflict of interest.

\appendix
\section{Proofs}

\normalsize
\subsection{Proof of Theorem \ref{thm:divg}}
\begin{proof}

We fix $\boldsymbol{\theta} := ((\boldsymbol{\mu}_1, \kappa_1), \ldots, (\boldsymbol{\mu}_p, \kappa_p)) \in \Theta^{p}$ such that $\boldsymbol{\mu}_l = \mathbf{x}_m$ for some pair $(l, m)$. We construct a sequence $(\boldsymbol{\theta}^{(q)}, \boldsymbol{\pi}^{(q)}), q=1,2,\ldots$ and show that $L({\cal X}; \boldsymbol{\theta}^{(q)}, \boldsymbol{\pi}^{(q)}) \uparrow \infty$ as $q \uparrow \infty$.
\\\\
For $k=1,2,\ldots,p$ and $q=1,2,\ldots$, we let $\boldsymbol{\mu}^{(q)}_{k} = \boldsymbol{\mu}_k$, and $\pi^{(q)}_{k} = (1 - 1/q) \pi_k + 1/(q p)$. It is easy to verify that $\sum_{k=1}^{p} \pi^{(q)}_k = 1$. For each $q$, we let $\kappa_l^{(q)} = q$, and for $k \ne l$, we let $\kappa_k^{(q)} = \kappa_k$. Since $\pi_k^{(q)} \ge 1 / (q p)$, the likelihood is lower bounded by:
$$ L({\cal X}; \boldsymbol{\theta}^{(q)}, \boldsymbol{\pi}^{(q)}) \ge \frac{1}{(q p)^{n}} \prod_{i=1}^{n} \sum_{k=1}^{p} f(\mathbf{x}_i; \boldsymbol{\mu}_k, \kappa_k^{(q)}) . $$
For the $m$th observation, we have
$$ \sum_{k=1}^{p} f(\mathbf{x}_m; \boldsymbol{\mu}_k, \kappa_k^{(q)}) \ge f(\mathbf{x}_m; \boldsymbol{\mu}_l, \kappa_l^{(q)}) = c_d(\kappa_l^{(q)}) e^{\kappa_l^{(q)}} .$$
For any $i \ne m$, and $h \ne l$, we have
$$ \sum_{k=1}^{p} f(\mathbf{x}_i; \mu_k, \kappa_k^{(q)}) \ge f(\mathbf{x}_i; \boldsymbol{\mu}_h, \kappa_h^{(q)}) = c_d(\kappa_h^{(q)}) e^{\kappa_h^{(q)} \boldsymbol{\mu}_h^{T} \mathbf{x}_i} .$$
Therefore, the likelihood function can be lower bounded by
\begin{eqnarray*}
 L({\cal X}; \boldsymbol{\theta}^{(q)}, \boldsymbol{\pi}^{(q)}) &\ge& \frac{1}{(qp)^{n}} c_d(\kappa_l^{(q)}) e^{\kappa_l^{(q)}} \prod_{i \ne n} c_d(\kappa_h^{(q)}) e^{\kappa_h^{(q)} \boldsymbol{\mu}_h^{T} \mathbf{x}_i} \\
& = & \frac{1}{(qp)^{n}} c_d(q) e^{q} \prod_{i \ne n} c_d(\kappa_h) e^{\kappa_h \boldsymbol{\mu}_h^{T} \mathbf{x}_i} 
.
\end{eqnarray*}
Since $c_d(\kappa) = {\cal O}(\kappa^{d/2 - 1/2})$, we have $L({\cal X}; \boldsymbol{\theta}^{(q)}, \boldsymbol{\pi}^{(q)}) \uparrow \infty$ as $q \uparrow \infty$.
\end{proof}

\subsection{Technical Lemmas}
The following technical lemmas are useful for the proof of the main results. Lemma \ref{bessel_lemma} gives the asymptotic expansion of the modified Bessel function of the first kind and is straight forward to derive.
\begin{lemma}
\label{bessel_lemma}
Let $I_r(\cdot)$ be the modified Bessel function of the first kind with degree $r$. As $z \rightarrow \infty$, we have
$$ I_r(z) = \frac{e^{z}}{\sqrt{2\pi z}} (1 + {\cal O}(z^{-1})) .$$
\end{lemma}
Lemma \ref{covering_lemma} concerns the covering of the surface of the unit hypersphere with $B_{\epsilon}$-balls and is needed for the proof of Lemma \ref{key_lemma_1}.
\begin{lemma}
\label{covering_lemma}
For any sufficiently small positive number $\epsilon$, there exists points $\boldsymbol{\eta}_1, \ldots, \boldsymbol{\eta}_m \in \mathbb{S}^{d-1}$ with $m \le A_1/\epsilon^{d-1}$ where $A_1 > 0$ is some constant which depends on the dimension $d$ such that for any $\mathbf{x} \in \mathbb{S}^{d-1}$, there exists $\boldsymbol{\eta}_j$ with $B_{\epsilon}(\mathbf{x}) \subset B_{2 \epsilon}(\boldsymbol{\eta}_j) $.
\end{lemma}

\begin{proof}

Fix $\epsilon > 0$ and consider an open cover $\{ B_{\epsilon}(\boldsymbol{\eta}_i)\}_i$ of $\mathbb{S}^{d-1}$. By compactness of $\mathbb{S}^{d-1}$, there exists a finite subcover $\{B_{\epsilon}(\boldsymbol{\eta}_1), \ldots, B_{\epsilon}(\boldsymbol{\eta}_m)\}$ of $\mathbb{S}^{d-1}$. Let $\mathbf{x} \in \mathbb{S}^{d-1}$ be fixed, and let $ \mathbf{z} \in B_{\epsilon}(\mathbf{x})$ be arbitrary. We must show that $d(\mathbf{z},\boldsymbol{\eta}_i) < 2 \epsilon$ for some $i$. 
\\\\
Since $\{B_{\epsilon}(\boldsymbol{\eta}_1), \ldots, B_{\epsilon}(\boldsymbol{\eta}_m)\}$ is an open cover of $\mathbb{S}^{d-1}$, we must have $\mathbf{x} \in B_{\epsilon}(\boldsymbol{\eta}_i)$ for some $i$. Therefore, 
$$ d(\mathbf{z}, \boldsymbol{\eta}_i) \le d(\mathbf{z},\mathbf{x}) + d(\mathbf{x}, \boldsymbol{\eta}_i) < \epsilon + \epsilon = 2\epsilon.$$
Hence, we have $B_{\epsilon}(\mathbf{x}) \subset B_{2 \epsilon}(\boldsymbol{\eta}_i)$. Since $x$ is arbitrary, the result follows. \\\\
The statement that $m \le A_1/\epsilon^{d-1}$ for some constant $A_1 > 0$ is clearly true for $d=2$, and the general case holds using proof by induction with a geometric argument.

\end{proof}

\subsection{Proof of Lemma \ref{key_lemma_1}}
\textbf{Proof of Lemma \ref{key_lemma_1}}
\begin{proof}
Let $\epsilon$ be a small positive number. By Lemma \ref{covering_lemma}, there exists $\boldsymbol{\eta}_1, \ldots, \boldsymbol{\eta}_m \in \mathbb{S}^{d-1}$ with $m \le A_1 / \epsilon^{d-1} $ such that for any $\mathbf{x} \in \mathbb{S}^{d-1}$, we have $B_{\epsilon}(\mathbf{x}) \subset B_{2 \epsilon}(\boldsymbol{\eta}_j)$ for some $j$.
Consequently, we have that 
\begin{eqnarray*}
\sup_{\boldsymbol{\mu} \in S^{d-1}} \bigg\{ \frac{1}{n} \sum_{i=1}^{n} I\big(X_i \in B_{\epsilon}(\boldsymbol{\mu})\big)  \bigg\} &\le& \max_{j=1,\ldots,m} \bigg\{ \frac{1}{n} \sum_{i=1}^{n} I\big(X_i \in B_{2 \epsilon}(\boldsymbol{\eta}_j)\big) \bigg\} \\
& \le & \max_{j=1,\ldots,m} \bigg\{ \frac{1}{n} \sum_{i=1}^{n} I \big( X_i \in B_{2\epsilon}(\boldsymbol{\eta}_j) \big) - \gamma_0(B_{2\epsilon}(\boldsymbol{\eta}_{j}))  \bigg\} \\
&& + \max_{j=1, \ldots, m} \bigg\{ \gamma_0(B_{2 \boldsymbol{\epsilon}}(\boldsymbol{\eta}_j))  \bigg\} ,
\end{eqnarray*}
where $ \gamma_0(B_{2\epsilon}(\boldsymbol{\eta}_j)) = \gamma_0(X \in B_{2\epsilon}(\boldsymbol{\eta}_j))$ is the probability that a random variable $X$ generated from the $\gamma_0$ takes value in the $2 \epsilon-$ball $B_{2\epsilon}$. For each $j=1, \ldots, m$, $\gamma_0(B_{2 \epsilon}(\boldsymbol{\eta}_j))$ can be bounded above by 
$$ \gamma_0(B_{2\epsilon}(\boldsymbol{\eta}_j)) = \int_{B_{2 \epsilon}(\boldsymbol{\eta}_j)} f_0(\mathbf{x}) d \omega(\mathbf{x}) \le M \omega(B_{2\epsilon}(\boldsymbol{\eta}_j)) = M A_2 \epsilon^{d-1} = \delta(\epsilon) , $$
where we recall that the constants $M$ and $A_2$ are defined in Equation \eqref{M_def} and \eqref{A2_def}, respectively, and the function $\delta(\cdot)$ is defined in Equation \eqref{delta_def}. This implies that
\begin{eqnarray}
\label{ineq1}
\max_{j=1, \ldots, m} \bigg\{ \gamma_0(B_{2\epsilon}(\boldsymbol{\eta}_j))  \bigg\} \le \delta(\epsilon) .
\end{eqnarray}
Define the quantities
$$ \Delta_{nj} := \bigg| \frac{1}{n} \sum_{i=1}^{n} I \big( X_i \in B_{2\epsilon}(\boldsymbol{\eta}_j) \big) - \gamma_0(B_{2\epsilon}(\boldsymbol{\eta}_{j})) \bigg| , \quad j=1,\ldots,m.$$
For $t > 0$, by Bernstein's inequality we have
\begin{eqnarray}
\label{ineq_bernstein}
 \mathbb{P}( \Delta_{nj} \ge t ) &\le& 2 \exp\bigg( - \frac{ \frac{1}{2} n^2 t^2 }{n \gamma_0(B_{2\epsilon}(\boldsymbol{\eta}_j) (1-\gamma_0(B_{2\epsilon}(\boldsymbol{\eta}_j))) + \frac{1}{3} nt }  \bigg) \nonumber \\
  & \le & 2 \exp \bigg( - \frac{n^2 t^2}{2 n M A_2 \epsilon^{d-1} + \frac{2}{3} n t }  \bigg) \nonumber \\
    &=& 2 \exp \bigg( - \frac{n t^2}{2\delta(\epsilon) + \frac{2}{3}  t } \bigg)
\end{eqnarray}
We note that $\delta(\epsilon) > \log n / n$ whenever $\epsilon^{d-1} > \log n / (M n A_2)$. Letting $t = \delta(\epsilon)$ in the inequality above, we obtain
$$ \mathbb{P}(\Delta_{nj} \ge \delta(\epsilon)) \le 2 n^{-3} .$$
Since $\epsilon^{d-1} > \log n / (M n A_2)$ implies $m < n$ for sufficiently large $n$, we apply the union upper bound to obtain 
\begin{eqnarray}
\label{ineq2}
 \mathbb{P} \Big( \max_j \Delta_{nj} \ge \delta(\epsilon) \Big) \le 2 n^{-2} .
\end{eqnarray}
Combining the two inequalities \eqref{ineq1} and \eqref{ineq2}, the first conclusion of the lemma follows by applying the Borel-Cantelli lemma.
\\\\
For the second statement of the lemma, we observe that $ 0 < \epsilon^{d-1} < \log n / (M n A_2)$ implies that $\delta(\epsilon) < \log n / n$. Let $t := (\log n)^{2} / n$, for large enough $n$, we have
$$ 2 \delta(\epsilon) < t / 3 .$$ 
Substituting $t$ into inequality \eqref{ineq_bernstein} gives
$$ \mathbb{P} \bigg( \Delta_{nj} \ge \frac{ (\log n)^{2} }{ n } \bigg) \le \exp( - (\log n)^{2} ) \le n^{-3} . $$
The second conlusion of the lemma follows from an application of the Borel-Cantelli lemma.

\end{proof}

\subsection{Proof of Strong Consistency of PMLE}
\begin{proof}
We prove the strong consistency of PMLE for the case of two mixture components. The proof of strong consistency for the general case of $p$ mixture components is analogous but significantly more tedious. We briefly outline the key steps for the proof of the $p$ mixture components case, which are also along the lines of Section 3.3 of \citep{Chen2008}.
\\\\
Recall that a two component mixture mixing measure has the form $ \gamma = \pi \gamma_{\theta_1} + (1 - \pi) \gamma_{\theta_2}$, where $0 \le \pi \le 1$, $\theta_1 = (\boldsymbol{\mu}_1, \kappa_1)$, $\theta_2 = (\boldsymbol{\mu}_2, \kappa_2)$ and the corresponding penalized log-likelihood function is given by 
$$ pl_n(\gamma) = l_n(\gamma) + \tilde{p}_n(\kappa_1) + \tilde{p}_n(\kappa_2) .$$ 
Let $K_0 = E_0 \log f(X; \gamma_0)$ where $E_0$ denotes the expectation under the true probability measure $\gamma_0$. We follow the strategy in \citep{Chen2008} to divide the parameter space $ \Gamma = \Pi \times \Theta^{2}$ into three regions $$ \Gamma_1 := \{\pi, 1 - \pi, (\boldsymbol{\mu}_1, \kappa_1), (\boldsymbol{\mu}_2, \kappa_2): \kappa_1 \ge \kappa_2 \ge \nu_0 \} ,$$ $$ \Gamma_2 := \{\pi, 1 - \pi, (\boldsymbol{\mu}_1, \kappa_1), (\boldsymbol{\mu}_2, \kappa_2): \kappa_1 \ge \tau_0, \kappa_2 \le \nu_0\} ,$$ $$\Gamma_3^ := \Gamma - (\Gamma_1 \cup \Gamma_2) .$$
We require $\nu_0$ and $\tau_0$ to be sufficiently large where the exact magnitude are to be specified later. We will show that the penalized MLE $\hat{\gamma}$ almost surely is not in $\Gamma_1$ or $\Gamma_2$. Therefore, $\hat{\gamma}$ must be in $\Gamma_3$ and its consistency follows from Theorem 5 of \cite{Redner1981}.
\\\\
We first consider the region $ \Gamma_1 $, and define the following index sets:
$$ D_1 := \{i: X_i \in B_{\epsilon_1}(\boldsymbol{\mu}_1)\}, \quad D_2 := \{i: X_i \in B_{\epsilon_2}(\boldsymbol{\mu}_2)\} ,$$
where $$   \epsilon_i = \frac{1}{(\log \kappa_i)^{2}}, \quad i=1,2 .$$ 
$D_1$ and $D_2$ consist of observations that are very close to $\boldsymbol{\mu}_1$ and $\boldsymbol{\mu}_2$, respectively. We separately assess the likelihood contributions of the observations in $D_1$ and $D_2$ in Lemmas \ref{Gamma1_lemma} and \ref{Gamma2_lemma}. 
\\\\
By Lemmas \ref{Gamma1_lemma} and \ref{Gamma2_lemma} the maximizer $\hat{\gamma}_n$ of $pl_n(\gamma)$ is almost surely in $\Gamma_3$. Lemmas \ref{Gamma1_lemma} and \ref{Gamma2_lemma} also imply that $\gamma_0 \in \Gamma_3$. We want to apply Theorem 5 of \cite{Redner1981} to conclude the strong consistency of the estimator $\hat{\gamma}_n$. First, it is clear following the definition of the metric $\rho$ on $\Theta$ given in Equation \ref{metric_theta} that $\Gamma_3$ is a compact subset of $\Gamma$ containing $\gamma_0$. For any $\theta = (\boldsymbol{\mu}, \kappa) \in \Theta$, we can choose $r$ sufficiently small such that for all $\theta' =(\boldsymbol{\mu}', \kappa')$ with $\rho(\theta, \theta') < r$, the density $f(\cdot; \boldsymbol{\mu}', \kappa')$ is bounded. Thus, we have
\begin{eqnarray*}
 \int_{ \mathbb{S}^{d-1} } \log \bigg( \max \big( 1, \sup_{ \theta': \rho(\theta, \theta') < r} f(\mathbf{x}, \boldsymbol{\mu}', \kappa') \big) \bigg) d \gamma_0 & \le &  \int_{\mathbb{S}^{d-1}} C d \gamma_0 \\
 & \le & C,
\end{eqnarray*}
for some constant $C$. Therefore, Assumption 2a of \cite{Redner1981} is also satisfied. For any $\theta_1, \theta_2 \in \Theta$ where $\theta_1 = (\boldsymbol{\mu}_1, \kappa_1)$ and $\theta_2 = (\boldsymbol{\mu}_2, \kappa_2)$, since $f(\cdot, \boldsymbol{\mu}_1, \kappa_1)$ is bounded, we have 
$$ \int_{\mathbb{S}^{d-1}}  |\log f(\mathbf{x}; \boldsymbol{\mu}_1, \kappa_1)| d \gamma_{\theta_2} < \infty .$$
Therefore, Assumption 4 of \cite{Redner1981} is also satisfied. We can conclude by applying Theorem 5 of \cite{Redner1981} that $\hat{\gamma}_n \rightarrow \gamma_0$ almost surely in the quotient space $\hat{\Gamma}$.
\\\\
We outline the key steps of the proof for the general case of $p$ mixture components, which follows the same strategy as the proof for the 2 components case.  We divide the parameter space $\Gamma = \Pi \times \Theta^p$ into $p+1$ regions
\begin{eqnarray*}
 \Gamma_1 := \{ (\pi_1, \ldots, \pi_p), (\boldsymbol{\mu}_1, \kappa_1), \ldots, (\boldsymbol{\mu}_p, \kappa_p): \kappa_1 \ge \kappa_2 \ge \cdots \ge \kappa_p \ge \nu_{10} \} ,\\
 \Gamma_k := \{ (\pi_1, \ldots, \pi_p), (\boldsymbol{\mu}_1, \kappa_1), \ldots, (\boldsymbol{\mu}_p, \kappa_p):  \kappa_1 \ge \cdots \ge \kappa_{p-k+1} \ge \nu_{k0}, \\ \nu_{(k-1)0} \ge \kappa_{p-k+2} \ge \cdots \ge \kappa_p \} ,
\end{eqnarray*}
for $k=2, \ldots, p$, and $\Gamma_{p+1} = \Gamma - \cup_{k=1}^{p} \Gamma_k $.\\
Given suitable constraints on the values of $\nu_{10}, \ldots, \nu_{k0}$, an extension of Lemma \ref{Gamma1_lemma} shows that $$ \sup_{\gamma \in \Gamma_1} pl_n(\gamma) - pl_n(\gamma_0) \rightarrow - \infty,\quad \mbox{a.s.} ,$$
and an extension of Lemma \ref{Gamma2_lemma} shows that 
$$ \sup_{\gamma \in \Gamma_k} pl_n(\gamma) - pl_n(\gamma_0) \rightarrow - \infty, \quad \mbox{a.s.} $$
for $k=2, \ldots, p$. Therefore, the probability that $\hat{\gamma}_n$, the maximizer of $pl_n(\gamma)$ belongs to $\Gamma_1, \ldots, \Gamma_p$ goes to zero. With $\hat{\gamma}_n$ almost surely in $\Gamma_p$ which is a compact subset of $\Gamma$, we can apply Theorem 5 of \cite{Redner1981} to conclude the strong consistency of $\hat{\gamma}_n$.

\end{proof}

\begin{lemma}
\label{Gamma1_lemma}
$ \sup_{\gamma \in \Gamma_1} pl_n(\gamma) - pl_n(\gamma_0) \rightarrow - \infty, \quad \mbox{a.s.} $
\end{lemma}

\begin{proof}

The log-likelihood contributions of observations in any index set $D$ is given by
$$ l_n(\gamma;D) = \sum_{i \in D} \log \bigg( \pi c_d(\kappa_1) \exp(\kappa_1 \mathbf{x}_i^{T} \boldsymbol{\mu}_1) + (1-\pi)  c_d(\kappa_2) \exp(\kappa_2 \mathbf{x}_i^{T} \boldsymbol{\mu}_2) \bigg) .$$
For any observation $i$ in $D_1$, its likelihood contribution is bounded above by $\exp(\kappa_1) / c_d(\kappa_1)$, and by the asymptotic expansion of the modified Bessel function of the first kind in Lemma \ref{bessel_lemma}, we have
\begin{eqnarray}
\label{ineq3}
  c_d(\kappa_1) \exp(\kappa_1)  \le A_3 \sqrt{\kappa_1} .
\end{eqnarray}
for some constant $A_3 > 0$. Consequently, the log-likelihood of observations in $D_1$ is bounded above by
$$ l_n(\gamma;D_1) \le n(D_1) \log( A_3 \sqrt{\kappa_1}) ,$$
where $n(D_1)$ is the number of observations in $D_1$. By Lemma \ref{key_lemma_2}, for 
$$ \frac{ \log n}{M n A_2} \le \epsilon_1^{d-1} < \xi_0 ,$$ 
$n(D_1)$ is almost surely bounded above by 
$$ n(D_1) = \sum_{i=1}^{n} I(X_i \in B_{\epsilon_1}(\boldsymbol{\mu}_1)) \le 4 n \delta(\epsilon_1) = 4 n M A_2 \epsilon_1^{d-1} .$$
Therefore, recalling that $\epsilon_1 = 1 / (\log \kappa_1)^{2}$, $l_n(\gamma;D_1)$ can be bounded above by:
\begin{eqnarray}
\label{D1_bound1}
 l_n(\gamma;D_1) &\le& 4 n M A_2 \epsilon_1^{d-1} ( \log A_3 \sqrt{\kappa} ) \nonumber \\ 
    & \le & A_4 n \frac{1}{(\log \kappa_1)^{2 d - 3}} \nonumber \\
    & \le & A_4 n \frac{1}{(\log \nu_0)^{2 d - 3}} 
\end{eqnarray}
where $A_4 > 0$ is some constant, and the last inequality follows from $\kappa_1 > \nu_0$. For
$$ 0 < \epsilon_1^{d-1} < \frac{ \log n }{ M n A_2} ,$$
we have $ n(D_1) \le 2 (\log n )^{2}$ almost surely by Lemma \ref{key_lemma_2}. Therefore, with condition C3 on the penalty function $\tilde{p}_n(\kappa_1)$, almost surely
\begin{eqnarray}
\label{D1_bound2}
 n(D_1) (\log A_3 \sqrt{\kappa_1}) + \tilde{p}_n(\kappa_1) \le 2 (\log n)^{2} (\log A_3 \sqrt{\kappa_1}) + \tilde{p}_n(\kappa_1) < 0 .
\end{eqnarray}
The two bounds \eqref{D1_bound1} and \eqref{D1_bound2} can be combined to form
\begin{eqnarray}
\label{lik_bound1}
 l_n(\gamma;D_1) + \tilde{p}_n(\kappa_1) \le A_4 n \frac{1}{ (\log \nu_0)^{2 d -3} } .
\end{eqnarray}
The same approach can be used to derive the same bound for observations in $D_1^{c} D_2$:
\begin{eqnarray}
\label{lik_bound2}
 l_n(\gamma;D_1^{c} D_2) + \tilde{p}_n(\kappa_2) \le A_4  n \frac{1}{ (\log \nu_0)^{2 d -3} } .
\end{eqnarray}
For any observation $\mathbf{x}$ that falls outside both $D_1$ and $D_2$, we have that $\arccos(\mathbf{x}^{T} \boldsymbol{\mu}_1) > \epsilon_1$ and $\arccos(\mathbf{x}^{T} \boldsymbol{\mu}_2) > \epsilon_2$. Using the Taylor expansion of $\cos(\epsilon)$ for positive $\epsilon$ around 0,  we can show that for $\mathbf{x} \in D_1^{c} D_2^{c}$, we have
$$ \mathbf{x}^{T} \boldsymbol{\mu}_i \le 1 - \frac{1}{3 \epsilon_i^{2}} = 1 - \frac{1}{3 (\log \kappa_i)^{4}}, \quad i = 1, 2 .$$ Consequently, recalling the inequality \eqref{ineq3}, and for large enough $\nu_0$, the log-likelihood contribution of such $\mathbf{x}$ is bounded above by
$$ \log A_3 + \log \sqrt{\kappa_i} - \frac{\kappa_i}{3 (\log \kappa_i)^{4}} \le \log \kappa_i - \frac{\kappa_i}{4 (\log \kappa_i)^{4}} \le  \log \nu_0 - \frac{\nu_0}{4 (\log \nu_0)^{4}} .$$
For large enough $n$, we must have $n(D_1^{c} D_2^{c}) \ge n / 2 $ almost surely. Therefore, almost surely the log-likelihood of the observations in $D_1^{c} D_2^{c}$ is bounded above by
\begin{eqnarray}
\label{lik_bound3}
 l_n(\gamma;D_1^{c}D_2^{c}) \le (n/2) \bigg( \log \nu_0 - \frac{\nu_0}{ (\log \nu_0)^{4} } \bigg) .
\end{eqnarray}
For sufficiently large $\nu_0$, the following inequalities hold
$$ 2 A_4 \frac{1}{(\log \nu_0)^{2d-3}} \le 1, $$
$$\log \nu_0 - \frac{\nu_0}{(\log \nu_0)^{4}} \le 2 K_0 - 4 .$$
Therefore, combining the bounds \eqref{lik_bound1}, \eqref{lik_bound2}, \eqref{lik_bound3}, the penalized log-likelihood can be bounded above by
\begin{eqnarray*}
  pl_n(\gamma) &\le& 2 A_4 n \frac{1}{(\log \nu_0)^{2d-3}} +  (n/2) \bigg( \log \nu_0 - \frac{\nu_0}{(\log \nu_0)^{4}}  \bigg) \\
     & \le & n + (n/2)(2 K_0 - 4) \\
     & = & n(K_0 - 1) .
\end{eqnarray*}
By strong law of large numbers, we have $n^{-1} pl_n(\gamma_0) \rightarrow K_0$ almost surely. Therefore,
$$ \sup_{\gamma \in \Gamma_1} pl_n(\gamma) - pl_n(\gamma_0) \rightarrow -\infty $$
almost surely.

\end{proof}

\begin{lemma}
\label{Gamma2_lemma}
 $ \sup_{\gamma \in \Gamma_2} pl_n(\gamma) - pl_n(\gamma_0) \rightarrow - \infty, \quad \mbox{a.s.} $
\end{lemma}
\begin{proof}
To establish a similar result for $\Gamma_2$, we define 
$$ g(\mathbf{x};\gamma) = \pi \exp \bigg(\frac{\kappa_1}{2} (\mathbf{x}^{T} \mu_1 - 1)\bigg) + (1-\pi)
c_d(\kappa_2) \exp\big( \kappa_2 \mathbf{x}^{T} \mu_2 \big) .$$
We note that the first part of the RHS above is not a vMF density, and is well defined as $\kappa_1 \rightarrow \infty$. Straightforward calculation shows that for all $\gamma \in \Gamma_2$, we have
$$ \int_{\mathbb{S}^{d-1}} g(\mathbf{x}; \gamma) dx < 1 .$$
Therefore, by Jensen's inequality, for all $\gamma \in \Gamma_2$,
\begin{eqnarray*}
 E_0 \bigg[ \log \frac{g(X;\gamma)}{f_0(X)} \bigg] < 0 ,
\end{eqnarray*}
where we recall that $f_0(\cdot)$ is the true density function. Since $\sup_{\gamma \in \Gamma_2} g(\mathbf{x};\gamma)$ is bounded and $g(\mathbf{x};\gamma)$ is continuous in $\gamma$ almost surely w.r.t. $f_0(\mathbf{x})$, it follows that
\begin{eqnarray}
\label{conv1}
 \sup_{\gamma \in \Gamma_2} \bigg\{ \frac{1}{n} \sum_{i=1}^{n} \log \Big( \frac{g(X_i;\gamma)}{f_0(X_i)} \Big) \bigg\} \rightarrow - \eta(\tau_0) < 0 , \quad \mbox{a.s.}
\end{eqnarray}
where $\eta(\tau_0) >  0 $ is an increasing function of $\tau_0$. Hence, we can find $\tau_0 > \nu_0$ such that
\begin{eqnarray}
\label{eta_ineq}
 A_4 \frac{1}{(\log \tau_0)^{2d-3}} \le \eta(\nu_0)/4 < \eta(\tau_0)/4 .
\end{eqnarray}
For any observation in $D_1$, its log-likelihood contribution is no larger than
$$ \log (A_3 \sqrt{\kappa_1} ) + \log g(\mathbf{x}; \gamma) .$$
For sufficiently large $\tau_0$, the log-likelihood contribution of any observation not in $D_1$ is no more than $\log g(\mathbf{x};\gamma)$. This follows since for large enough $\kappa_1 > \tau_0$, 
\begin{eqnarray*}
 c_d(\kappa_1) e^{\kappa_1 \mathbf{x}^{T} \mu_1} &=& c_d(\kappa_1) e^{\kappa_1 (\mathbf{x}^{T} \boldsymbol{\mu}_1 - 1)} e^{\kappa_1} \\ 
       & \le & A_3 \sqrt{\kappa_1}  e^{\kappa_1 (\mathbf{x}^{T} \boldsymbol{\mu}_1 - 1)} \\
      & \le &  e^{ \frac{\kappa_1}{2} (\mathbf{x}^{T} \boldsymbol{\mu}_1 - 1) } .
\end{eqnarray*}
Therefore, the penalized log-likelihood is almost surely bounded above by
\begin{eqnarray*}
 & & \sup_{\Gamma_2} pl_n(\gamma) - pl_n(\gamma_0)  \\
 &\le & \sup_{\kappa_1 \ge \tau_0} \bigg\{ \sum_{i \in D_1} \log(A_3 \sqrt{\kappa_1}) + \tilde{p}_n(\kappa_1) \bigg\}
  + \sup_{\Gamma_2} \bigg\{ \sum_{i=1}^{n} \log \frac{g(X_i;\gamma)}{f_0(X_i)} \bigg\} + p_n(\boldsymbol{\kappa}_0)  \\
  & \le & A_4 n \frac{1}{(\log \tau_0)^{2d-3}} - \frac{3}{4} \eta(\tau_0) n + p_n(\boldsymbol{\kappa}_0)  \\
  & \le & -\frac{1}{2} \eta(\tau_0) n +p_n(\boldsymbol{\kappa}_0) \rightarrow -\infty
\end{eqnarray*}
where $\boldsymbol{\kappa}_0$ is the vector of the concentration parameters of the true measure $\gamma_0$, the second inequality follows from \eqref{lik_bound1} and \eqref{conv1}, and the last inequality follows from \eqref{eta_ineq}.

\end{proof}

\bibliographystyle{asa}
\bibliography{main}

\end{document}